\documentclass{llncs}
\usepackage{arash}
\usepackage{verbatim}
\usepackage{enumitem}
\usepackage{wrapfig}
\newcommand{\SHORT}[1]{}

\usepackage{graphics}
\usepackage{epsfig}
\usepackage{amsmath,amssymb}
\newcommand{\IGNORE}[1]{}
\renewcommand{\log}{\lg}
\newcommand{\Patrascu}{P\v{a}tra\c{s}cu}

\newenvironment{description*}%
  {\vspace{-1ex}\begin{description}%
    \setlength{\itemsep}{-0.5ex}%
    \setlength{\parsep}{0pt}}%
  {\end{description}}
\newenvironment{itemize*}%
  {\vspace{-1ex}\begin{itemize}%
    \setlength{\itemsep}{-0.5ex}%
    \setlength{\parsep}{0pt}}%
  {\end{itemize}}
\newenvironment{enumerate*}%
  {\vspace{-1ex}\begin{enumerate}%
    \setlength{\itemsep}{-0.5ex}%
    \setlength{\parsep}{0pt}}%
  {\end{enumerate}}

\begin{document}

\title{Succinct Indices for Range Queries with applications to Orthogonal Range Maxima\thanks{Work done while Farzan was employed by, and Raman was visiting, MPI.}}
\author{Arash Farzan\inst{1}
\and J. Ian Munro\inst{2}
\and Rajeev Raman\inst{3}}

\institute{Max-Planck-Institut f\"ur  Informatik, Saarb\"ucken, Germany.\and
University of Waterloo, Canada.\and
University of Leicester, UK.}

\maketitle

\begin{abstract}
We consider the problem of preprocessing $N$ points in 2D, each endowed with a priority, to answer
the following queries: given a axis-parallel rectangle, determine the point with the largest
priority in the rectangle.  Using the ideas of the \emph{effective entropy} of range maxima queries
and \emph{succinct indices} for range maxima queries, we obtain a
structure that uses $O(N)$ words and 
answers the above query in $O(\log N \log \log N)$ time.  This is a 
direct improvement of Chazelle's result from 1985 \cite{Chazelle88} for this problem -- 
Chazelle required $O(N/\epsilon)$ words to answer queries in $O((\log N)^{1+\epsilon})$ time for
any constant $\epsilon > 0$. 
\end{abstract}

%ARXIV VERSION

\section{Introduction}
\label{sec:intro}

%range searching is important
Range searching is one of the most fundamental problems in computer science with important
applications in areas such as computational geometry, databases and string processing.
The input is a set of $N$ points in general position in $\mathbb{R}^d$ (we focus on the case $d=2$),   
where each point is associated with \emph{satellite} data, and an aggregation function
defined on the satellite data.  We wish to preprocess the input
to answer queries of the following form efficiently: given any 2D axis-aligned rectangle $R$,
return the value of the aggregation function on the satellite data of  
all points in $R$. 
%There exists a great deal of literature on range searching with a variety of
%aggregate functions spanning over a few
%decades~\cite{Lueker78,Bentley80,McCreight85,Chazelle88,Chazelle901,Chazelle902,Alstrup00,Hellerstein02,JaJa04,Agarwal04,Nekrich07,Nekrich09,Karpinski09,afshani:dominance,Chan10,Chan11} 
%max range searching is an important subpart
Researchers have considered range searching 
with respect to diverse aggregation functions such as  emptiness checking, counting, reporting,
minimum/maximum, etc. \cite{encyclopaedia}. 
In this paper, we consider the problem of \emph{range maximum} 
searching (the minimum variant is symmetric), where the satellite data associated with
each point is a numerical \emph{priority}, and the aggregation function
is ``arg max'', i.e.,  we want to report the point with the maximum priority in the given
query rectangle. This aggregation function is \emph{the} canonical one to study, among
those of the ``commutative semi-group'' class \cite{Gabow1984,Chazelle88}.

%space important: polylog of chazzelle to almost logarithmic
Our primary concern is the space requirement of the data structure ---  
we aim for \emph{linear-space} data structures, namely those
that occupy $O(N)$ words --- and seek to minimize query time subject 
to this constraint.    The space usage is a fundamental concern
in geometric data structures due to very large data volumes; indeed, space usage is a main
reason why range searching data structures like {quadtrees}, which have poor
worst-case query performance, are preferred in many
practical applications over data structures such as {range trees}, which have
asymptotically optimal query performance.   Space efficient solutions 
to range searching date to the work of Chazelle \cite{Chazelle88} over a quarter century ago, and
Nekrich \cite{Nekrich} gives a nice survey of much of this work.
Recently there has been a flurry of activity on various aspects of space-efficient
range reporting, and for some aggregation functions there has even been
attention given to the constant term within the space usage \cite{BoseHMM09,Nekrich}.

We now formalize the problem studied by our paper, as well as those of
\cite{Chazelle88,Chan10,Karpinski09}. 
We assume input points are in \emph{rank space}: 
%and furthermore, there are 
%no duplicate $x$ or $y$-coordinates. In other words, 
the $x$-coordinates of the $n$ points are $\{0,\ldots,N-1\} = [N]$,
and the $y$-coordinates are given by a permutation
$\upsilon: [N] \rightarrow [N]$, such that the points are
$(i,\upsilon(i))$ for $i = 0,\ldots,N-1$.  The priorities of
the points are given by another permutation $\pi$ such
that $\pi(i)$ is the priority of the point $(i, \upsilon(i))$.
The reduction to rank space  can be performed in $O(\log N)$ 
time with a linear space structure even if the original and query points 
are points in $\mathbb{R}^2$ \cite{Gabow1984,Chazelle88}.  The query rectangle
is specifed by two points from $[N] \times [N]$ and includes
the boundaries (see Fig.~\ref{fig:examples}(R)).  
Analogous to previous work, we also assume %our structures also work in 
the word-RAM model with word size $\thetah{\log N}$ bits\footnote{$\log x = \mbox{\rm log}_2 x$}.
%and use RAM operations on integers of $\oh{\log N}$ bits.
%
\begin{figure}
\centerline{\includegraphics[scale=0.33]{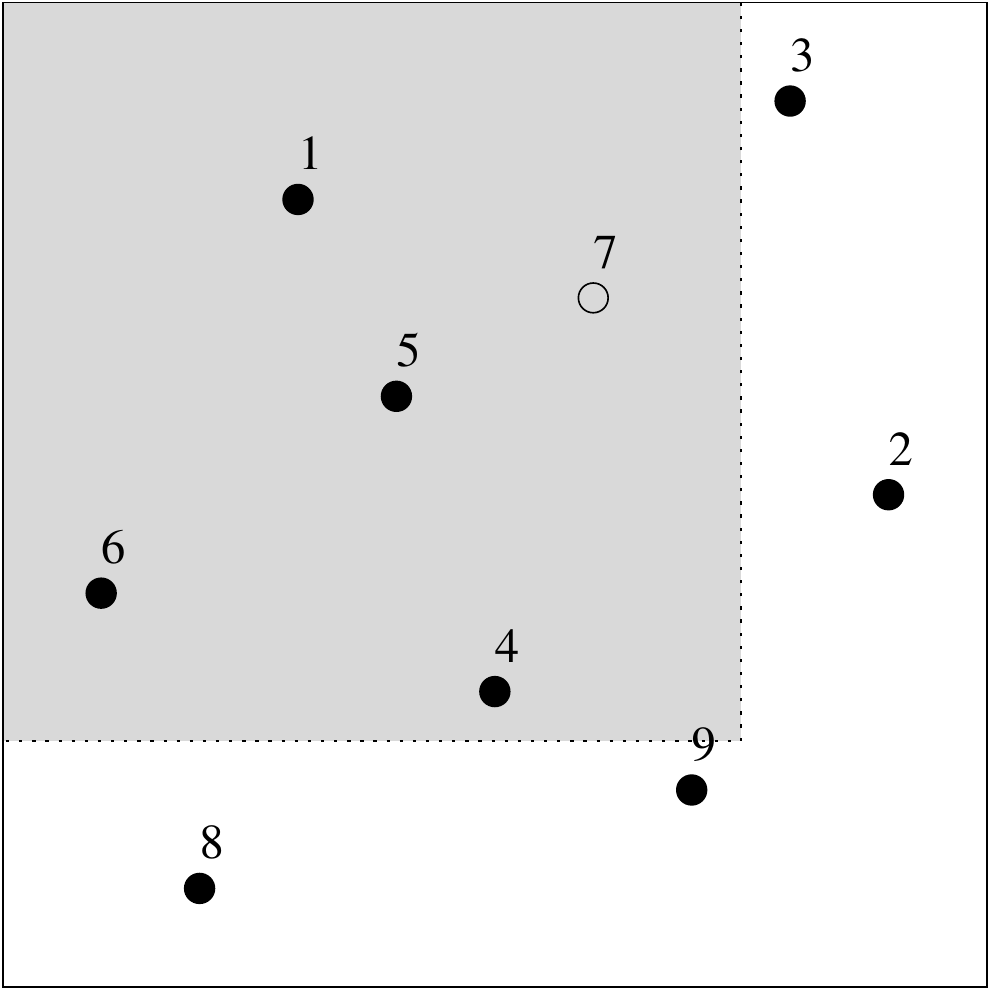}~~~~~\includegraphics[scale=0.33]{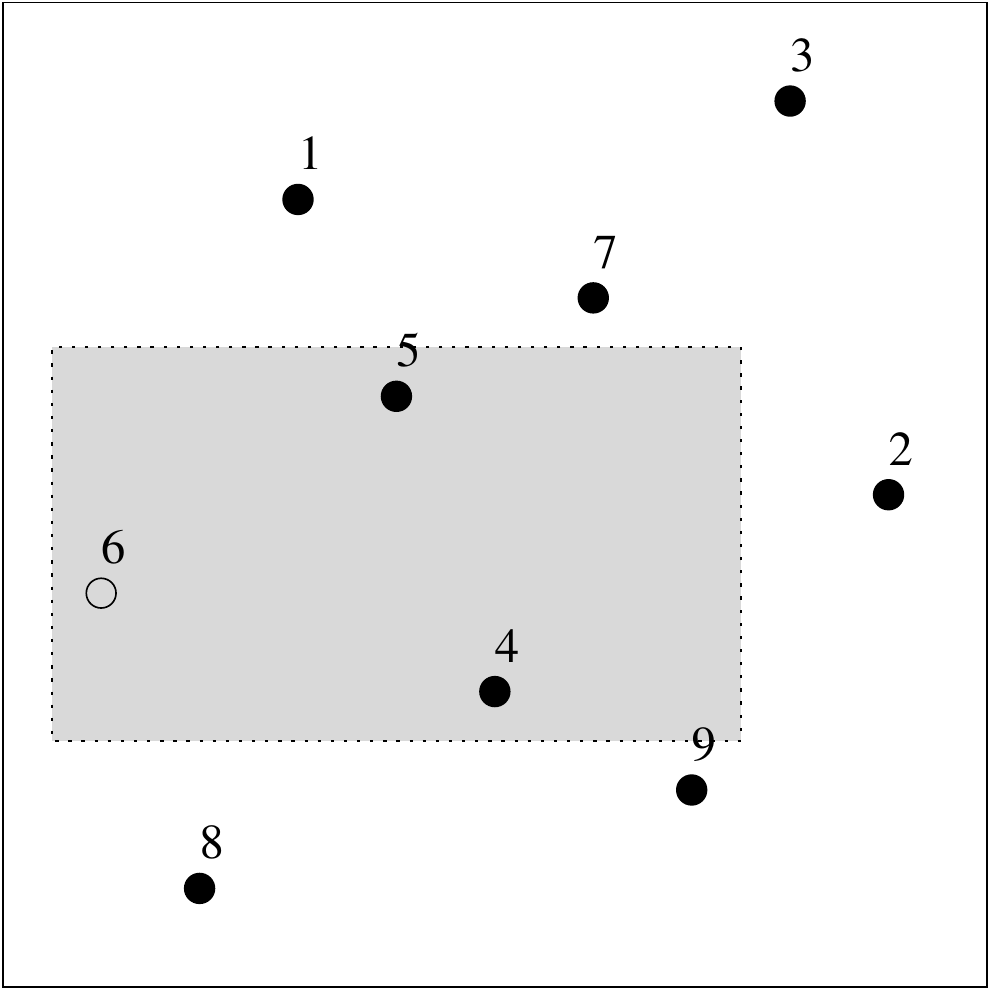}}
\caption{2-sided and 4-sided range maximum queries.  The numbers with the points represent
their priorities, and the unshaded points are the answers.}
\label{fig:examples}
\end{figure}

Range maximum searching is a well-studied problem (see Table~\ref{tab:prevWork}).
Chazelle~\cite{Chazelle88} gave a few space/time tradeoffs covering a broad spectrum. 
To the best of our knowledge, the solution with the lowest query time that
uses only $O(N)$ words is still that of Chazelle~\cite{Chazelle88}, who gave a query time that is
polylogarithmic in $N$. More precisely, he gave a data structure 
of size  $O(\frac{1}{\epsilon} N)$ words with query time
$O(\log^{1+\epsilon} N)$ for any fixed $\epsilon > 0$.  
Other recent results on the range maximum problem are as follows.
Karpinski~\etalcite{Karpinski09} studied the problem of 3D five-sided range emptiness
queries which is closely related to range maximum searching in 2D. As observed
in~\cite{Chan10}, their solution yields a query time of $(\log \log N)^{O(1)}$
with an index of size
$N (\log\log N)^{O(1)}$ words. Chan~\etalcite{Chan10} currently give the best query time of
$\oh{\log\log N}$, but this is at the expense of using $\oh{N\log^{\epsilon} N}$ words,
for any fixed $\epsilon > 0$. However, 
there has been no improvement in the running time for linear-space data structures.
In this paper, we improve Chazelle's long-standing result by giving a data structure of $O(N)$ 
words and  reducing the query time from polylogarithmic to 
``almost" logarithmic, namely, $\oh{\log N \log \log N}$.
\begin{table}[t]
\small
\begin{center}
\begin{tabular}{|l|c|c|}\hline
Citation & Size (in words) & Query time \\ \hline \hline
Chazelle'88~\cite{Chazelle88} & $\oh{N\log^{\epsilon}N}$ & $\oh{\log N}$ \\ \hline
Chan et al.'10~\cite{Chan10} & $\oh{N\log^{\epsilon}N}$ & $\oh{\log\log N}$ \\ \hline
Karpinski et al.'09~\cite{Karpinski09} & $\oh{N(\log\log N)^{O(1)}}$ & $\oh{(\log\log N)^2}$ \\ \hline
Chazelle'88~\cite{Chazelle88} & $\oh{N\log\log N}$ & $\oh{\log N\log\log N}$ \\ \hline
Chazelle'88~\cite{Chazelle88} & $\oh{\frac{1}{\epsilon}N}$ & $\oh{\log^{1+\epsilon} N}$ \\ \hline
\textbf{NEW} & $\oh{N} $ & $\oh{\log N\, \log\log N }$ \\\hline
\end{tabular}
\end{center}
\caption{Space/time tradeoffs for 2D range maximum searching in the word RAM.  \label{tab:prevWork}}
\end{table}
Although our primary focus is on 4-sided queries,
which specify a rectangle that is bounded from
all sides,  we also
need to consider 2-sided and 3-sided queries, which are ``open'' on two and one
side respectively (thus a 2-sided query is specified by a single point $(i,j)$---see Fig.~\ref{fig:examples}(L)---and
a 3-sided query by two points $(i,j)$ and $(k,l)$ where either $i=k$ or $j=l$).  
Our solution recursively divides the points into horizontal and 
vertical slabs, and a query rectangle is decomposed into smaller  2-sided, 3-sided,
and 4-sided queries. A key intermediate result %which gives us tremendous space-efficiency 
is the data structure for 2-sided queries.  The 2-sided sub-problems 
are partitioned into smaller sub-problems, which are stored
in a ``compressed'' format that is then ``decompressed'' on demand.
The ``compression'' uses the idea that to   
answer 2-sided range maxima queries on a problem of size $m$, 
one need not store the entire total order of priorities using 
$\thetah{m\log m}$ bits: $\oh{m}$ bits suffice, i.e., the \emph{effective entropy} 
\cite{Golinetal2011} of 2-sided queries is low.
This does not help immediately, 
since $\thetah{m\log m}$ bits are needed to store the coordinates
of the points comprising these sub-problems: to overcome this bottleneck we use ideas from
\emph{succinct indices} \cite{Barbayetal} --- namely, we separate the storage
of point coordinates from the data structure for answering range-maximum queries.
The latter data structure does not store point coordinates but instead obtains
them as needed from a global data structure.

We solve 3-sided and 4-sided subqueries by recursion, or by using
structures for range maximum queries on matrices \cite{AtallahYuan,Brodaletal}. 
When recursing, we cannot afford the space required
to store the structures for rank space reduction for each such subproblem: a further key idea 
is to  use a single global structure to map references to points within 
recursive subproblems back to the original points.

By reusing ideas from the data structure for 2-sided queries, 
we obtain two stand-alone results on \emph{succinct indices} for 2-sided range maxima queries.  
In a succinct index, we wish to answer queries on some data, and it
is assumed that the succinct index is not charged for the space required to
store the data itself, but only for any additional space it requires to
answer the queries.  However, the succinct index can access the data only
through a specific interface (see \cite{Barbayetal} for a discussion of the 
advantages of succinct indices).  In our case, given $N$ points in rank 
space, together with priorities, we wish to answer 2-sided range-maximum
queries under the condition that the point coordinates are stored
``elsewhere''---the data structure is not ``charged'' for the space needed
to store the points ``elsewhere''---and are assumed to be available to the data structure in one of two ways:

\begin{itemize}
\item Through an orthogonal range reporting query.  Here, we assume that a query
that results in $k$ points being reported takes $T(N,k)$ time.  We assume that
$T$ is such that $T(N, O(k)) = O(T(N,k))$.
\item As in Bose et al. \cite{Boseetal09}, we assume that the point coordinates are stored
in read-only memory, permitting random access to the coordinates of the $i$-th point.
However, the ordering of the points is specified by the data structure, which we call
the \emph{permuted-point} model.\footnote{Clearly, a succinct index for this problem is interesting
only if it uses $o(N)$ words = $o(N \log N)$ bits of space, so if the points are
stored in arbitrary order in read-only memory, the succinct index cannot itself
store the permutation that re-orders the points.}
\end{itemize}
In both cases, we are able to achieve $O(N)$-bit indices with fast query time, namely
 $O(\log \log N \cdot (\log N  + T(N,\log N)))$ time, and $O(\log \log N)$ time, respectively.

The paper is organized as follows.  We first describe some building blocks used in
Section 2.  Section 3 is devoted to our main result, and Section 4 describes the
succinct index results.  

%Hence, we say the \emph{effective entropy} 
%\cite{Golinetal2011} of 2-sided
%range maxima queries is $\thetah{n}$ bits.  
%rank space reduction

\section{Preliminaries} \label{sec:prelim}
Our result builds upon a number of relatively new results, all
of which are essential in one way or the other to obtain the final bound.
In order to support mapping between recursive sub-problems, we
use the following primitives on a set $S$ of $N$ points in rank space.
A \emph{range counting} query reports 
the \emph{number} of points within a query rectangle:
\begin{lemma}[\cite{JaJa04}]\label{lem:orc}
Given a set of $N$ points in rank space in two dimensions, there is a data structure 
with $O(N)$ words of space that supports range counting queries in $\oh{\log N / \log\log N}$ time.
\end{lemma}
A \emph{range reporting} structure supports the operation of 
listing the coordinates of all points within a query rectangle. 
We use the following consequence of a result of Chan \etalcite{Chan10}:
\begin{lemma}[\cite{Chan10}]\label{lem:orr}
Given a set of $N$ points in rank space in two dimensions, there is a data structure with $O(N)$
words of space that supports range reporting queries in $\oh{(1+k)\log^{1/3} N}$ time where $k$ is
the number of points reported.
\end{lemma}
The \emph{range selection} problem is as follows: given an input
array $A$ of size $N$, to preprocess it so that given a query
$(i,j,k)$, with $1 \le i \le j \le N$,  we return an index $i_1$ 
such that $A[i_1]$ is the $k$-th smallest of the elements
in the subarray $A[i],A[i + 1],\ldots,A[j]$. 
\begin{lemma}[\cite{DBLP:conf/isaac/BrodalJ09}]\label{lem:rangemedian}
Given an array of size $N$, there is a data structure with
$O(N)$ words of space that supports range selection queries in
$O(\log N/\log \log N)$ time.
\end{lemma}
\section{The data structure}
\newcommand{\Inf}{\mathit{Inf}}
\newcommand{\Left}{\mathit{Left}}
\newcommand{\Right}{\mathit{Right}}

In this section we show our main result:
\begin{theorem}
\label{thm:4sfinal}
Given $N$ points in two-dimensional rank space, and their priorities, there is
a data structure that occupies $O(N)$ words of space and
answers range maximum queries in $O(\log N \log \log N)$ time.
\end{theorem}
We first give an overview of the data structure.
We begin by storing all the points (using their input coordinates)
once each in the structures of Lemmas~\ref{lem:orc} and \ref{lem:orr}. 
We also store an instance of the data structure of Lemma~\ref{lem:rangemedian} 
once each for the
arrays $X$ and $Y$, where $X[i] = \nu(i)$ and $Y[i] = \nu^{-1}(i)$ for
$i \in N$ ($X$ stores the $y$-coordinates of the points in
order of increasing $x$-coordinate, and $Y$ the $x$-coordinates in 
order of increasing $y$-coordinate). These four ``global'' data structures
use $O(N)$ words of space in all.

We recursively decompose the problem
\`{a} la Afshani \etalcite{afshani:dominance}.
Let $n$ be the recursive problem size (initially $n = N$).
Given a problem of size $n$, we divide the problem into $n/k$ mutually
disjoint horizontal slabs of size $n$ by $k$, and $n/k$ mutually disjoint
vertical slabs of size $k$ by $n$. A horizontal and vertical slab intersect 
in a square of size $k\times k$. We recurse on each horizontal or vertical slab:
observe that each horizontal or vertical slab has exactly $k$ points in it,
and is treated as a problem of size $k$---i.e. it is logically comprised
of two permutations $\upsilon$ and $\pi$ on $[k]$ (Fig.~\ref{fig:cases}(L); Sec. \ref{sec:space}).  
Clearly, given a slab in a problem of size $n$ containing $k$ points, 
we require some kind of mapping between the coordinates in the slab (which in one dimension
will be from $[n]$) and the
recursive problem in order to view the slab
as a problem of size $k$.  A key to the space-efficiency of our 
data structure is that this mapping is \emph{not} explicitly stored, and
is achieved through a \emph{slab-rank} operation.  The 
 \emph{slab-select} problem is a generalization of the inverse 
%where a point expressed in the coordinate system of the recursive problem
%needs to be mapped back to the coordinate system of the top-level
of the slab-rank problem (Sec. \ref{sec:slabrank}).

%In order that the original query is mapped appropriately to the recursive problem,
%we need to compute the mapping of the query to the size-$k$ recursive problem:
%such a mapping is called 
The given query rectangle is decomposed into a number of {disjoint} 
%1-sided, 
2-sided, 3-sided
and 4-sided queries, based upon the decomposition of the input into slabs.  
There are three kinds of \emph{terminal} queries which do not generate further
recursive problems: all 2-sided queries are terminal,
4-sided queries whose sides are slab boundaries  
are terminal, and all queries that reach slabs at the bottom
of the recursion are terminal.  The problems (or data structures)
that involve answering
terminal queries (or are used to answer terminal queries) are also termed terminal. 
Each terminal query produces 
some \emph{candidate} points: the set of all
candidate points must contain the final answer.   A key invariant needed
to achieve the space bound is that 
all terminal problems of size $n$---except those at the bottom of
the recursion---use space $o(n \log n)$ bits (Sec. \ref{sec:space}).

Clearly, since storing the input permutation representing the point sets in terminal
problems takes $\Theta(n \log n)$ bits, we do not store them explicitly.  Instead, 
the terminal data structures are \emph{succinct indices}---the points that comprise them are 
accessed by means of queries to a single global data structure.  The terminal 4-sided
problems use an index due to \cite{Brodaletal}, while the 2-sided problems reduce
the range maximum query to planar point location. 
Although there is a succinct index for planar point location \cite{Boseetal09},
this essentially assumes that the points that comprise the planar subdivision
are stored explicitly (``elsewhere'') and permuted in a manner specified by
their data structure. In our case, if the points are represented explictly, they would occupy
$\Theta(N \log N / \log \log N)$ words across all recursive problems.  
A key to our approach is an implicit representation of the planar sub-division in the
recursive problems,
relevant parts of which are recomputed from a ``compressed'' representation 
at query time (Sec.~\ref{sec:2sided}); the other key step is to note that
$O(n)$ bits suffice to encode the priority information needed to answer 2-sided
queries in a problem of size $n$ (Sec.~\ref{sec:effentropy}).

\subsection{A recursive formulation and its space usage}
\label{sec:space}

The recursive structure is as follows.  Let $L = \log N$, and consider
a recursive problem of size $n$ (at the top level $n = N$).   
%In what follows, all logarithms are to base 2 unless otherwise stated.  
We assume that $N$ is a power of 2, as are a number of expressions which
represent the size of recursive problems.  This can readily be achieved
by replacing real-valued parameters $x \ge 1$ by $2^{\lfloor \log x \rfloor}$
or $2^{\lceil \log x \rceil}$ without affecting the asymptotic complexity.
Unless we have reached the bottom of the recursion, we partition the 
input range $[n] \times [n]$ into mutually disjoint vertical slabs 
of width $k = \sqrt{nL}$ and also into
mutually disjoint horizontal slabs of height $k = \sqrt{nL}$ -- each such
slab contains $k$ points and can be logically viewed as a recursive
problem of size $k$ (see Fig.~\ref{fig:cases}(L)).  
Observe that the input is divided into 
$(n/k)^2 = n/L$ \emph{squares} of size $k\times k$, each representing
the intersection of a horizontal slab with a vertical one.  
We need to answer either 
2-sided, 3-sided or 4-sided queries on this problem.

The data structures associated with the current recursive problem are:
\begin{itemize}
\item for problems at the bottom of the recursion, we store an instance of 
Chazelle's data structure which uses $\oh{n \log n}$ bits of space
and has query time $\oh{(\log n)^2}$.
\item For 2-sided queries (which are terminal) in non-terminal problems 
we use the data structure with space usage $O(n \sqrt{\log n})$ bits 
 described in Sec.~\ref{sec:2sided}.
\item 3- and 4-sided queries, all of whose sides are slab boundaries, are \emph{square-aligned}:
they exactly cover a rectangular sub-array of squares.  For such queries,
we use a $O(n)$-bit data structure comprising.
\begin{itemize}
\item A $n/k \times n/k$ matrix  containing the
(top-level) $x$ and $y$ coordinates and priority of the maximum point (if any) 
in each square. This uses $O(n/L \cdot L) = O(n)$ bits.
\item The data structure of \cite{AtallahYuan,Brodaletal} for answering
2D range maximum queries on the elements in the above matrix.  This also uses $O(n)$ bits.
\end{itemize}
\end{itemize}
%
%For 1-sided (terminal) queries we only need to store $O(n)$ bits ---
%$n$ for each direction -- which indicate the columns (rows) containing
%horizontal (vertical) prefix maxima. 
Finally, each recursive
problem has $O(\log N) = O(L)$ bits of ``header'' information, 
containing, e.g., the bounding box of the problem in the top-level
coordinate system.  Ignoring the header information,
% However, there are only $O(\log N)$ recursive problems overall,
%so this is negligible. Thus, 
the space usage is given by:
$$S(n) = 2 \sqrt{n/L} S(\sqrt{nL}) + O(n\sqrt{\log n}),$$
which after $r$ levels of recursion becomes:
$$S(N) = 2^r \frac{N^{1-1/2^r}}{L^{1-1/2^r}} S(N^{1/2^r} L^{1-1/2^r}) + 
\oh{2^r N \sqrt{\log (N^{1/2^r} L^{1-1/2^r})}}.$$
The recursion is terminated for the first level $r$ where 
$2^r \ge \log N / \log \log N$.  At this level, the problems
are of size $\oh{(\log N)^2}$ and $\omegah{(\log N)^{1.5}}$
and the second term in the space usage
becomes $O(N \log N)$ bits.  Applying $S(n) = O(n \log n)$ for the
base case, we see that the first term is
$O((\log N/\log \log N) \cdot N \cdot \log \log N) = O(N \log N)$ bits,
and the space used by the header information is indeed negligible.

\begin{figure}
\begin{centering}
\includegraphics[scale=0.305]{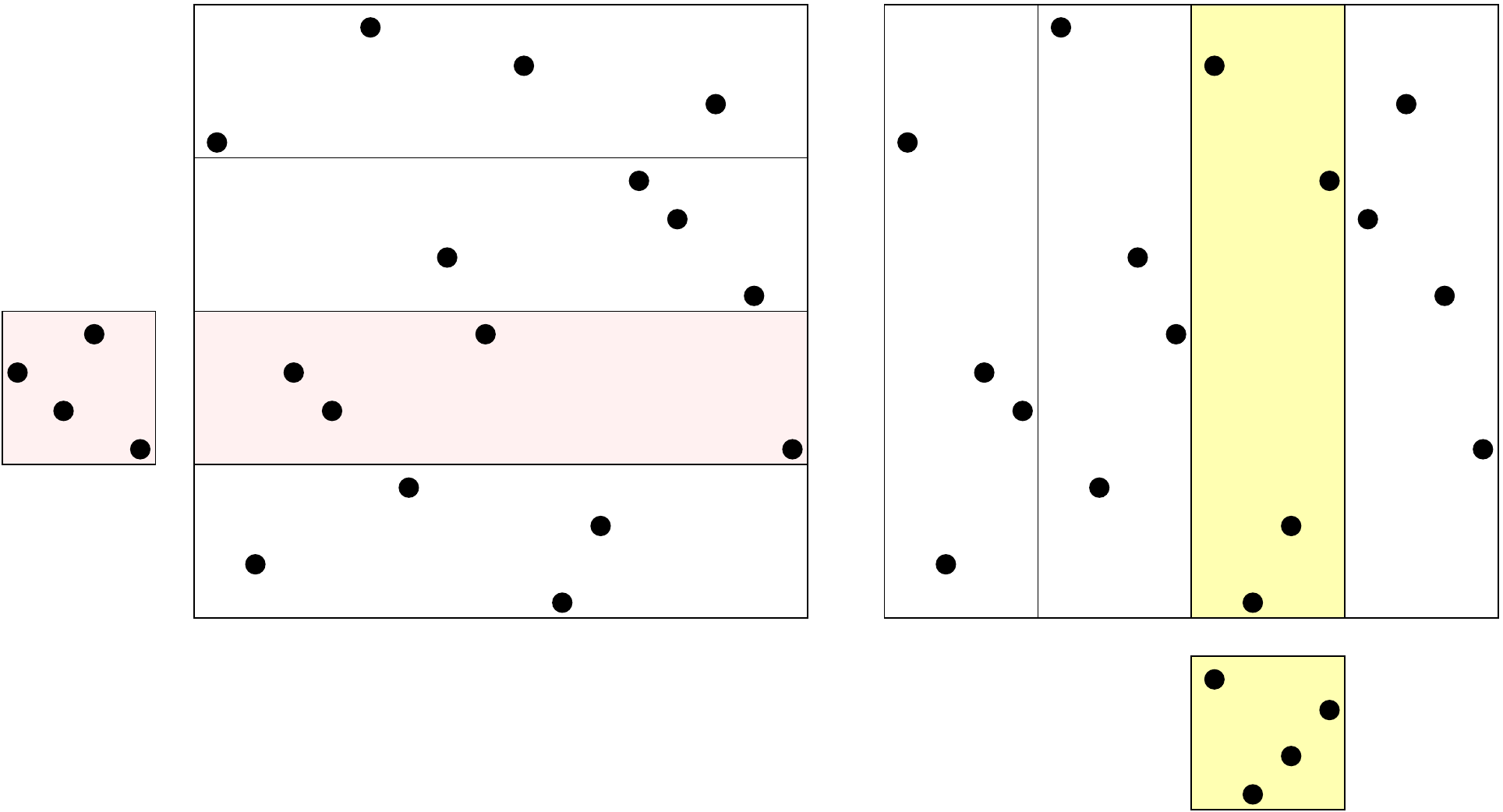}
~~~~~~~~~~~~~
\raisebox{0.8cm}{\includegraphics[scale=0.6]{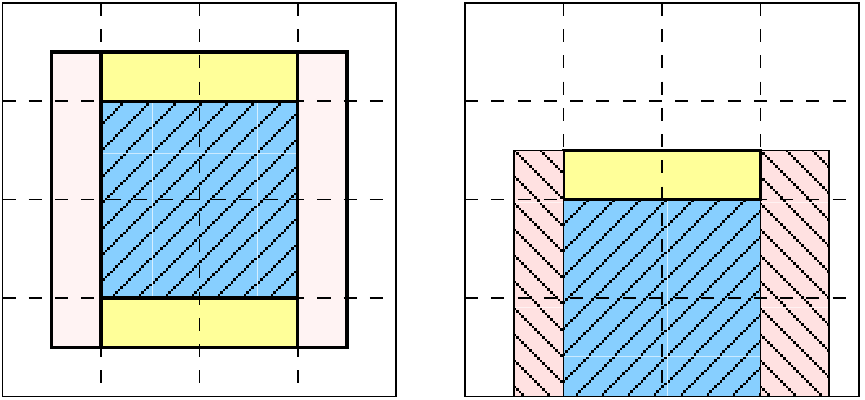}}
\end{centering}
\caption{The recursive decomposition of the input (L) and queries (R). In (R),  
shaded problems are terminal problems.  The 4-sided query is decomposed into a
square-aligned 4-sided query in the middle and four recursive 3-sided queries, two
in horizontal slabs and two in vertical slabs.  The 3-sided query is decomposed into 
two 2-sided queries in vertical slabs, a square-aligned 3-sided query and one recursive
3-sided query in a horizontal slab.}
\label{fig:cases}
\end{figure}

We now discuss the time complexity.  In general, we have to answer either
2-sided, 3-sided or 4-sided queries on a slab. Note that (see Fig.~\ref{fig:cases}(R)):
\begin{itemize}
\item A 2-sided query is terminal and generates one candidate.
\item A 3-sided query results in at most one recursive 3-sided query on a slab (generating
no candidates) at most two 2-sided queries on slabs, and at most one square-aligned
3-sided query (generating one candidate).
\item A 4-sided query either results in a recursive 4-sided query on a slab (generating 
no candidates) or generates at most one square-aligned 4-sided query (generating one
candidate), plus up to four 3-sided queries in slabs.
\end{itemize}
Since each 3-sided query only generates one recursive 3-sided query, the number
of recursive problems solved, and hence the number of candidates, is $O(r) = O(\log \log N)$.
  The time complexity will depend on the cost
of mapping query points and candidates between the problems and their
recursive sub-problems and the cost of solving the terminal problems; the former
is discussed next.

\subsection{The slab-rank and slab-select problems}
\label{sec:slabrank}

The input to each recursive problem of size $n$ is given in local
coordinates (i.e. from $[n]\times [n]$). Upon decomposing the query to
this problem, we need to solve the following \emph{slab-rank} problem
(with a symmetric variant for vertical slabs):

\begin{quote} Given a point $p = (i^*,j^*)$ in top-level coordinates,
which is mapped to $(i,j)$ in a recursive problem of size $n$, such
that $(i,j)$ that lies in a horizontal slab of size $n \times k$, 
map $(i,j)$ to the appropriate position $(i',j')$
in the size $k$ problem represented by this slab. 
\end{quote}
We formalize the ``inverse'' \emph{slab-select} problem as follows:
\begin{quote} Given a rectangle $R$ in the coordinate system of a recursive problem,
return the top-level coordinates of all points that lie within $R$.
\end{quote}
The following lemma assumes and builds upon the four ``global'' data structures
mentioned after the statement to Theorem~\ref{thm:4sfinal}. 
%For simplicity,
%we do not give the result for slab-select for general $R$ but focus on the useful case:
%
\begin{lemma}\label{lem:slab-rank-select} The slab-rank problem can be solved
in $O(\log N/\log \log N)$ time, and the slab-select problem in 
$O(\log N/\log \log N)$ time as well, 
provided that $R$ contains at most $\oh{\sqrt{\log N}}$ points. 
\end{lemma}

\begin{proof}
We first consider the slab-rank problem.  Without loss of generality, assume that the
given $(i,j)$ is in a horizontal $n \times k$ slab. To translate $j$ to $j'$, we 
only need to subtract the appropriate multiple of $k$ in $O(1)$ time.
To map $i$ to $i'$, we need to count the number of points in the slab with $x$-coordinate
smaller than $i$ (see Fig.~\ref{fig:slabrank}(L)).  Since the top-level coordinates of 
the point $(i,j)$ are known, and the top-level coordinates of the slab are also
stored in the ``header'' of the slab by assumption, top-level coordinates of all
sides of the query are known.  As all input points are stored in a global
instance of the data structure of Lemma~\ref{lem:orc}, counting the number
of points can be performed by an orthogonal range counting query in  
$O(\log N/\log \log N)$ time.
%
%
%The fourth side
%is easily determined as it is the vertical line
%passing through $i$.  The global coordinates of $i$ are known, as $i$ is
%either a vertex of the original input query, or is the result of 
%intersecting a horizontal (or vertical) line through the vertex of
%a recursive problem (whose top-level coordinates are inductively known)
%with a vertical (or horizontal) line defining a recursive sub-problem. 
%

\begin{figure}
\centerline{\includegraphics[scale=0.5]{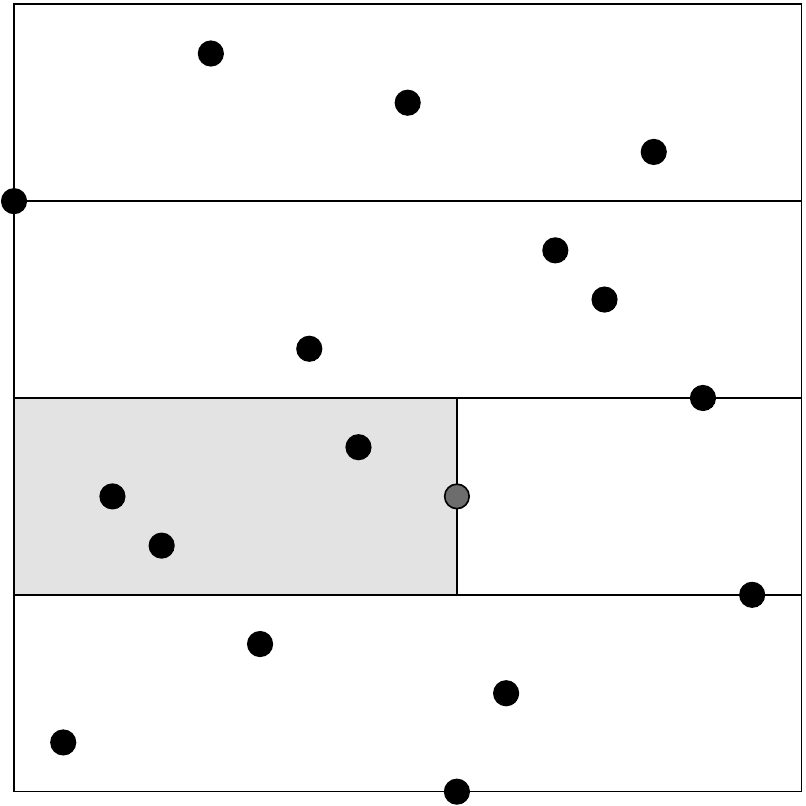}~~~~~~~~\includegraphics[scale=0.48]{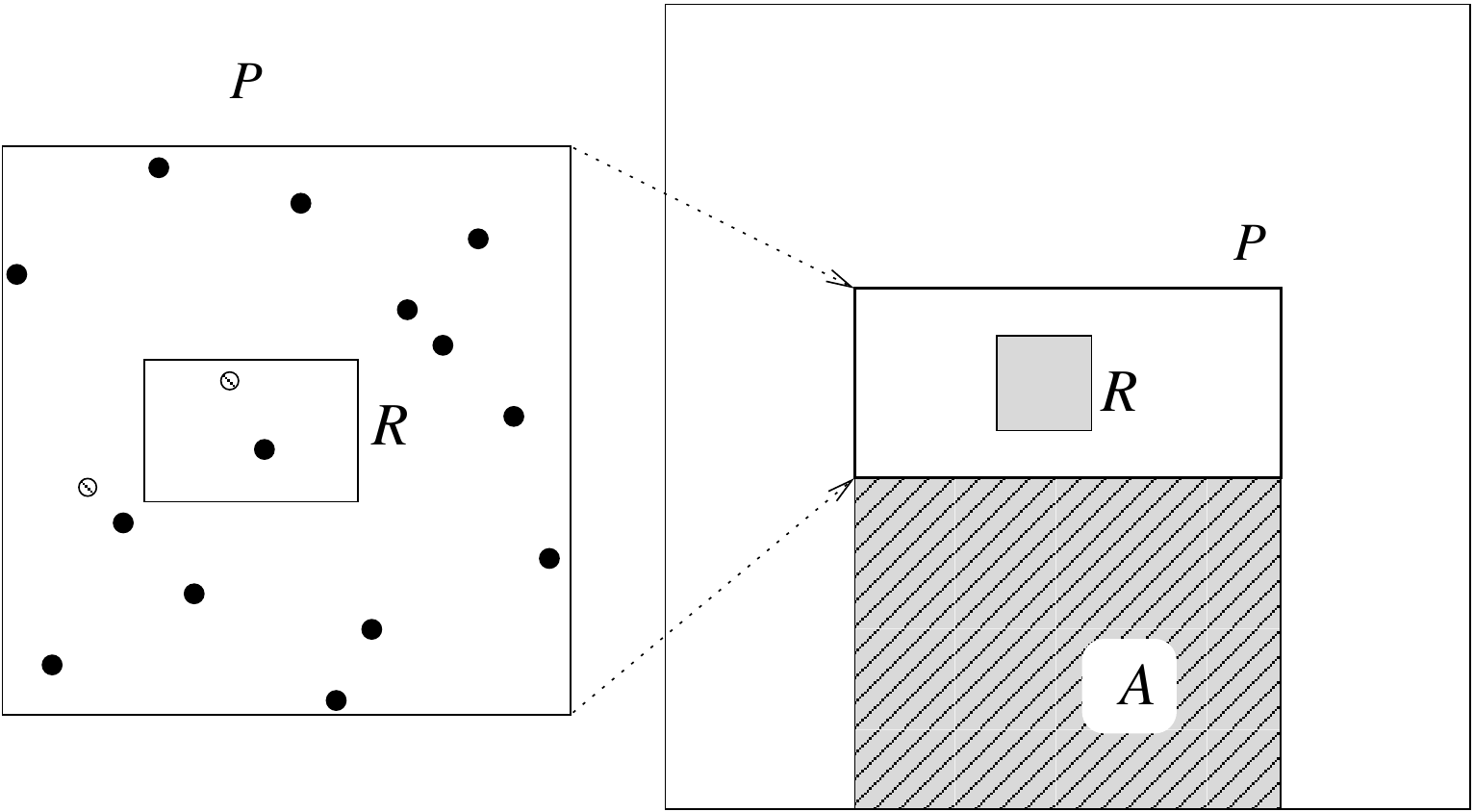}}
\caption{Slab-rank (left) and slab-select (right).}
\label{fig:slabrank}
\end{figure}
% and 
%instead ensure that given a (candidate) 
%point in a recursive problem, we return the top-level coordinates of
%a rectangle that contains $O(\sqrt{\log N})$ points and that also
%contains the candidate. 

The slab-select problem is solved in a similar manner, and is
most easily explained with reference to Fig.~\ref{fig:slabrank}(R).  We are
given a recursive sub-problem $P$ and a rectangle $R$ within $P$, and we
know $P$'s bounding box in the top-level problem (as shown on the far right).
Our aim is to retrieve the top-level coordinates of the image of $R$ in the
top-level problem.  Suppose that the local $x$ and $y$ coordinates of $R$ 
are $x_l$, $x_r$, $y_b$ and $y_t$.
Then we perform a orthogonal range count (Lemma~\ref{lem:orc})
in the area $A$, which lies under $P$
but within $P$'s $x$-coordinates. This takes $O(\log N/\log \log N)$ time, and
let $z$ be the value returned. 
We then select the $z+y_b$ and $z+y_t$-th smallest $y$-coordinates within 
$X[x_l], \ldots, X[x_r]$ by Lemma~\ref{lem:rangemedian}, also taking $O(\log N/\log \log N)$ time.  The (top-level)
$y$-coordinates of the points returned (shown shaded in the middle)
are $R$'s boundaries in the top-level
coordinate system.  A similar query on $Y[y_b],\ldots,Y[y_t]$ gives the
other boundaries of $R$ in the top-level coordinate system. The $\oh{\sqrt{\log N}}$
points in this rectangle are then retrieved in $\oh{(\log N)^{5/6}} = o(\log N/\log \log N)$ time
by Lemma~\ref{lem:orr}.

\qed{}
\end{proof}

\begin{remark} 
In all applications of the slab-rank result, the top-level coordinates of
$(i,j)$ will be known, as $(i,j)$ will either be a vertex of the original 
input query, or is the result of intersecting a horizontal (or vertical) 
line through the vertex of a recursive problem (whose top-level coordinates 
are inductively known) with a vertical (or horizontal) line defining a 
recursive sub-problem. 
\end{remark}

\subsection{Encoding 2-sided queries}
\label{sec:effentropy}

In this section we show Lemma~\ref{lem:2sentropy}.
Although the reduction of 2-sided range maxima 
queries at point $q$ ($RMQ(q)$ hereafter) to orthogonal
planar point location is not new, the observation
about the amount of priority information needed to answer 
$RMQ$ is new (and essential).
This lemma shows that although storing the permutation $\pi$ itself requires
$\Theta(n\log n)$ bits, the ``effective entropy'' \cite{Golinetal2011} of the
permutation with respect to 2-sided range maximum queries 
is much lower, generalizing the equivalent statement regarding 1D range-maximum
queries \cite{DBLP:journals/siamcomp/FischerH11,Vuillemin1980}.

%\begin{figure}
%\centerline{{\includegraphics[scale=0.40]{example-inf.pdf}}~~~~~~~~~~~~~~~~{\includegraphics[scale=0.40]{2s-lb.pdf}}}
%\caption{Example illustrating Lemma~\protect{\ref{lem:2sentropy}} (left).  The horizontal lines %
%are the lines of influence. Vertical dotted lines show where a point has terminated %
%the line of influence of another point.  The arrow shows how point location in %
%the lines of influence answers the 2-sided query with lower right hand corner at $q$, returning %
%the point with priority 7.  Example illustrating the lower bound (right).}
%\label{fig:infexample}
%\end{figure}
%
%
\begin{lemma}
\label{lem:2sentropy}
Given a set $S$ of $n$ points from $\mathbb{R}^2$ and relative
priorities given as a permutation $\pi$ on $[n]$, 
the query $RMQ(q)$ can be reduced to point
location of $q$ in a collection of at most $n$ horizontal 
semi-open line segments,
whose left endpoints are points from $S$, and 
whose right endpoints have $x$-coordinate 
equal to the $x$-coordinate of
some point from $S$. Further, given at most 
$2(n-r) + \log {{n} \choose {r}} \le 3n$ bits of extra information, 
the collection of line segments can be 
reconstructed from $S$, where $r$ is the number of
\emph{redundant} points---those that are never the 
answer to any query---in $S$, and this bound is tight.
\end{lemma}
\begin{proof}
Assume the points are in general position and that the
2-sided query is open to the top and left. Associate 
each point $p = (x(p), y(p)) \in S$ with a horizontal semi-open 
\emph{line of influence}, possibly of length zero,
whose left endpoint (included in the line) is $p$ itself,
and is denoted by $\Inf(p)$, and contains all
points $q$ such that $y(q) = y(p)$, $x(q) \ge x(p)$ and $RMQ(q) = p$.
It can be seen that (see e.g. \cite{DBLP:journals/ipl/MakrisT98}) the answer to
$RMQ(q)$ for any $q \in \mathbb{R}^2$ can be obtained by 
shooting a vertical ray upward from $q$ until 
the first line $\Inf(p)$ is encountered;
the answer to $RMQ(q)$ is then $p$ (if no line is
encountered then there is no point in the 2-sided
region specified by $q$).  See Fig.~\ref{fig:infexample} for an example.

\begin{figure}
\centering
\includegraphics[width=0.45\textwidth]{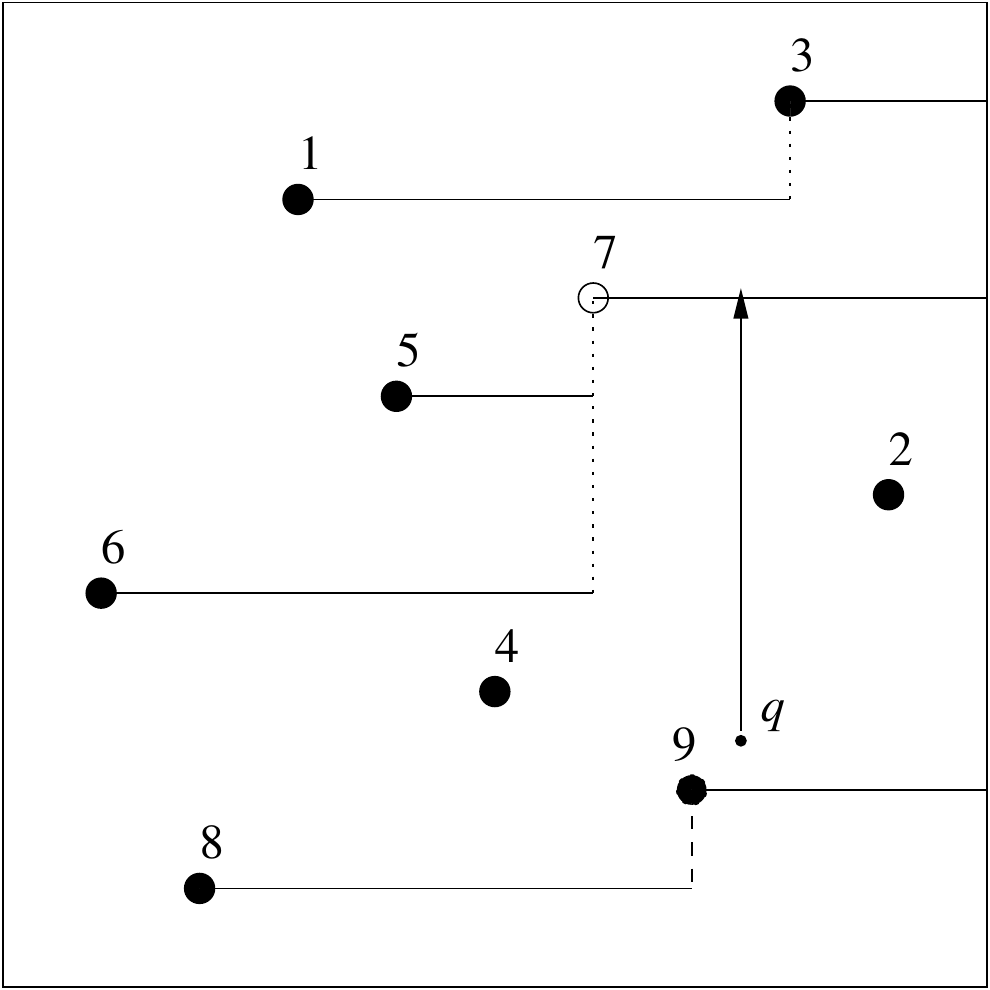}
~~~~
\includegraphics[width=0.45\textwidth]{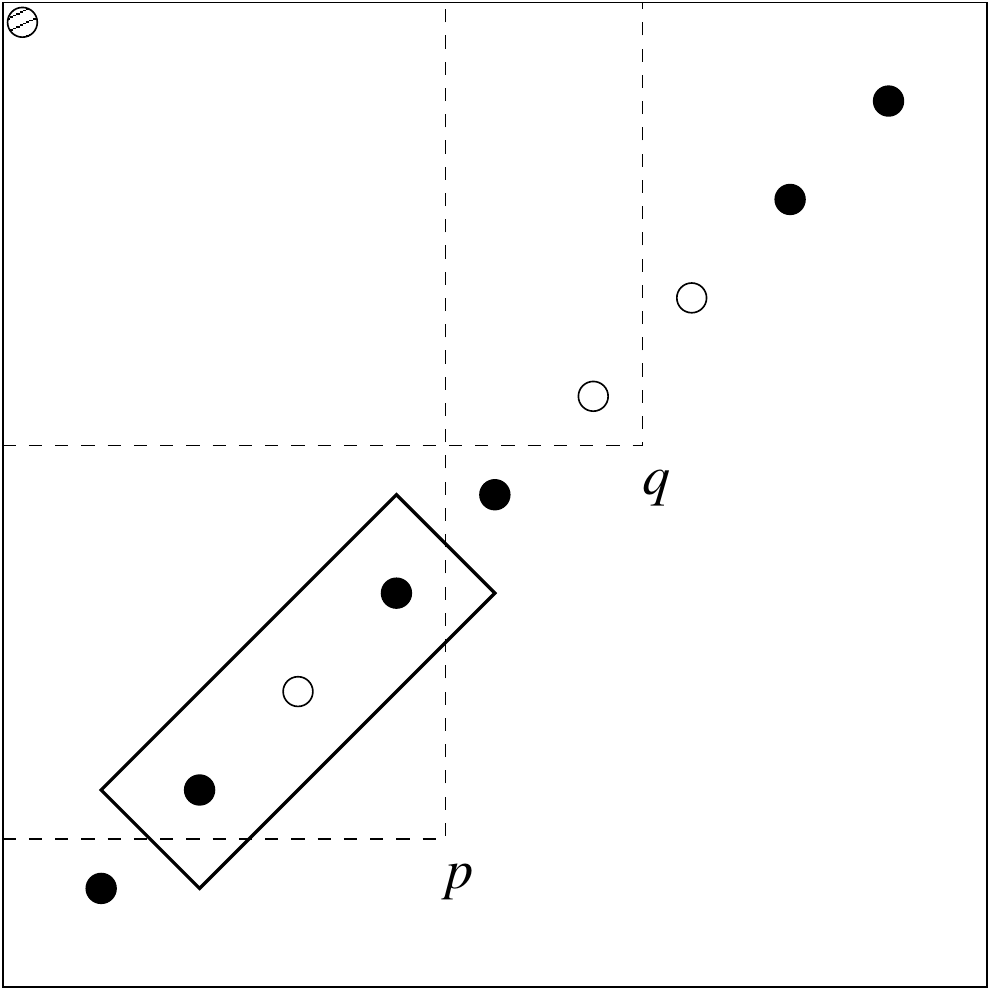}

\caption{(Left) Example for Lemma~\protect{\ref{lem:2sentropy}}.  The horizontal lines are the lines of influence. Vertical dotted lines show where a point has terminated the line of influence of another point. The arrow shows how point location in the lines of influence answers the 2-sided query with lower right hand corner at $q$, returning the point with priority 7. (Right) Example showing tightness of space bound --- points along the diagonal are $A_\ell$ and the queries at $p$ and $q$ illustrate how 1D RMQ queries can be answered and redundancy of elements tested, respecitvely.}
\label{fig:infexample}
\end{figure}

The set 
$\Inf(S) = \{\Inf(p) | p \in S\}$ can be computed by sweeping 
a vertical line from left to right.  At any given position
$x = t$ of the sweep line, the sweep line will intersect 
$\Inf(S')$ for some set $S'$ (initially $S' = \emptyset$).
If $S' = p_{i_1},\ldots,p_{i_r}$ such that
$y(p_{i_1}) > \ldots > y(p_{i_r})$ then it follows
that $\pi(i_1) < \ldots < \pi(i_r)$ (the current lines
of influence taken from top to bottom represent points
with increasing priorities).  Upon reaching the next
point $p_s$ such that $y(p_{i_j}) < y(p_s) < y(p_{i_{j+1}})$,
either (i) $\pi(s) < \pi(i_j)$---in this case
$\Inf({p_s})$ is empty---or (ii) $\pi(k) > \pi(i_j)$.
In the latter case, it may be that $\pi(k) > \pi(i_{j+1}), \ldots
\pi(i_{j+k})$ for some $k \ge 0$, which would mean that
$\Inf(p_{i_{j+1}}),\ldots,\Inf(p_{i_{j+k}})$
are terminated, with their right endpoints being $x(p_s)$.
To construct $\Inf(S)$, therefore,
only $O(n)$ bits of information are needed: for each point,
one bit is needed to indicate whether case (i) or (ii) holds,
and in the latter case, the value of $k$ needs to be stored.
However, $k$ can be stored in unary using $k+1$ bits, and
the total value of $k$, over the course of the entire sweep,
is at most $(n-r)$, giving a total of at most $2(n-r)$ bits.
The bit-string that indicates whether case (i) or (ii) holds
can be stored in 
$\left \lceil \log {{n} \choose {n-r}}\right \rceil \le n$ bits.

To show that this is tight, 
we consider the \emph{1D range maximum with redundant entries} problem,
defined as follows.  Given an array $A$ of size $n$, of which $r$ entries are
redundant, answer the following two queries:  firstly,
given an index $i$, state whether $A[i]$ is redundant, and 
secondly, given indices $1 \le i \le j \le n$, return
the index of the largest value in a query interval $A[i],\ldots,A[j]$
(ignoring redundant values).  Given $n, r$, it is easy to see that
$2 (n-r) + \log {{n} \choose {r}} - O(\log n)$ bits are required to
encode the answers to the above queries, namely, to answer these
queries without accessing $A$.  This is because there are
${{n}\choose{r}}$ choices for the positions of the redundant
elements, and for each choice of the positions of the redundant
elements there are $C_{n-r} = \frac{1}{n-r+1}{{2(n-r)}\choose{n-r}}$  partial 
orders among the values in $A$ that can be distinguished by 1D range maximum queries \cite{DBLP:journals/siamcomp/FischerH11,Vuillemin1980}.  
We associate the values in $A$ with
a set of $n$ points $A_\ell$ placed on a slanted line 
(in increasing $x$-coordinate order), and give each point in $A_\ell$
a priority equal to the corresponding entry in $A$.  
The point associated with each redundant value in $A$ is given a priority of
$-\infty$.  In addition we place an $n+1$st point $z$ that dominates
all of the points in $A_\ell$ (and hence is included in any query that also
includes a point from $A_\ell$), whose priority is greater than $-\infty$ but
less than the smallest priority associated with a non-redundant entry
of $A$ (see Fig.~\ref{fig:infexample}(R)). The 1D range maximum 
query problem with redundant entries can can be solved via 2-sided 2D range 
maximum queries that include only $z$ and the appropriate sub-range of
$A_\ell$. Also, we can determine if an entry in $A$ is redundant by making
a 2-sided query that includes only the corresponding point from $A_\ell$
and $z$; the point is redundant iff the answer to such a query is
$z$. \qed{} \end{proof}

\begin{remark} It can be shown that 
$2(n-r) + \log {{n} \choose {r}} \le n \log 5 + o(n)$, which is
at most $2.33 n$ bits for large $n$.
\end{remark}

\subsection{Data structures for 2-sided queries}
\label{sec:2sided}

In this section we show the following lemma:
\begin{lemma}
\label{lem:2sfinal}
Given a recursive sub-problem of size $n$, we can answer 2-sided queries on this problem
in $O(\log N)$ time using $O(n \sqrt{\log n})$ bits of space.
\end{lemma}
This lemma assumes and builds upon the four ``global'' data structures
metioned after the statement to Theorem~\ref{thm:4sfinal}. 
We view the given problem of size $n$, on which we need to support 2-sided queries,
as point location in a collection of $O(n)$ horizontal line segments as in Lemma~\ref{lem:2sentropy}, but with a limited space budget of $O(n \sqrt{\log n})$ bits.
We therefore need to devise
an implicit representation of these problems. We begin with an overview of the process.
Let $T$ be the set of $n$ points in the
sub-problem we are considering. We start with an explicit representation
of $\Inf(T)$ and choose a parameter
$\lambda = \thetah{\sqrt{\log n}}$.  We select $O(n/\lambda)$
lines of influence that partition the plane into 
rectangular \emph{regions} with $O(\lambda)$ points 
(from $T$) and parts of line segments (from $\Inf(T)$) \cite{BCR02,CP09} and store 
a standard point location data structure on the selected lines of influence.
This data structure, called the \emph{skeleton}, requires 
$O((n/\lambda) \log n) = \oh{n \sqrt{\log n}}$ bits.  
Furthermore, we store $O(\lambda)$
bits of information with each region (including the
$O(\lambda)$-bit encoding of priority information from
Lemma~\ref{lem:2sentropy}).

The query proceeds as follows.  Given a query point $q$, we first perform
a point location query on the skeleton to determine the region $R$ 
in which $q$ lies.  We now need to reconstruct the original point
location structure within $R$, and perform a slab-select to determine
the points of $T$ that lie within this region. This, together with
the priority information, allows us to partially---but not fully, since
lines of influence may originate from outside $R$---reconstruct the point location structure within $R$.  To handle
lines of influence starting outside $R$, we do a binary search with
$O(\log \lambda)$ steps, where in each step we need to perform a 
slab-select, giving the claimed bound. The details are as follows.
%
%This is
%done in a series of steps. 
%First, we perform a ``slab-select'' query, starting with 
%the given region, to determine a set of $O(\sqrt{\log N})$
%points from the original set $S$ that include the points in this region.
%We then use binary search (each step of which involves a slab-rank)
%to prune this set down to the points
%that lie within this region.  This allows us to construct 
%enough of the arrangement of line segments within the region
%to permit point location among lines originating within the
%region.  A binary search (again potentially involving slab-selects
%and pruning) completes the process.  We now describe each step in detail.

\subsubsection*{Preprocessing.}
%Suppose that we are given a
%recursive sub-problem of size $n$ whose set of points
%is denoted by $T$. 
As noted above, we first create $\Inf(T)$, take $\lambda = \sqrt{\log n}$ and select a set of
points $T' \subseteq T$ with the following properties: (a) $|T'| = O(n/\lambda)$;
(b) the \emph{vertical decomposition}, whereby we shoot vertical rays upward and 
downward from each endpoint of each segment in $\Inf(T')$ until they hit another segment 
(see Fig.~\ref{fig:verticaldecomp}), of the plane
induced by $\Inf(T')$\footnote{Note that the extent of a line segment in $\Inf(T')$ is defined, as originally, wrt points in $T$, and not wrt the points in $T'$.} decomposes the plane into
$O(n/\lambda)$ rectangular regions each of which has at most
$O(\lambda)$ points from $T$ and parts of line segments 
from $\Inf(T)$ in it.
$T'$ always exists and can be found by plane sweep
\cite[Section 3]{BCR02},\cite[Section 4.3]{CP09}.
The skeleton is any standard point location data structure on $\Inf(T')$.

Let $R$ be any region, and let
$\Left(R)$ ($\Right(R)$) be the set of line segments from
$\Inf(T)$ that intersect the left (right) boundaries of
$R$, and let $P(R)$ be the set of points from
$T$ in $R$.  We store the following bit strings for $R$:

\begin{enumerate}

\item For each line segment $\ell \in \Left(R)$, ordered top-to-down
by $y$-axis, a bit that indicates whether the right endpoint
of $\ell$ is in $R$ or not; similarly for $\ell \in \Right(R)$, 
a bit indicating whether $\ell$ begins in $R$ or not.

\item If the left boundary of $R$ is adjacent to other regions
$R_1,R_2,\ldots$ (taken from top to bottom) and $l_i \ge 0$ 
represents the number of line segments from $\Left(R)$ that
also intersect $R_i$, then we store a bit-string
$0^{l_1}10^{l_2}1\ldots$.  A similar bit-string is stored for
the right boundary of $R$.

\item For each point in $P(R)$ and each line segment in 
$\Left(R)$, a bit-string of length $|P(R)| + |\Left(R)|$ 
whose $i$-th bit indicates 
whether the $i$-th largest $y$-coordinate in $P(R) \cup \Left(R)$ 
is from $P(R)$ or $L(R)$.
\end{enumerate}
The purpose of (1) and (2) is to trace a line segment as it crosses multiple regions: if a line segment crosses from a region $R'$ to a region
$R''$ on its right, then given its position in $\Right(R')$, we
can deduce its position in $\Left(R'')$ and vice-versa. However, there is no
useful bound on the number of regions a single line segment
may cross, so we store the following information:

\begin{enumerate}
\setcounter{enumi}{3}
\item Suppose that a line segment $\ell = \Inf(p)$ for some $p \in T$ 
crosses $m \ge \lambda$ regions. Then, in every $\lambda$th region 
that $\ell$ crosses, we explicitly store the region containing $p$, 
and $p$'s local coordinates.  As in (1), for each region $R$, we
store one bit for each $\ell\in \Right(R)$, $\ell = \Inf(p)$ , indicating whether
or not $R$ holds information about $p$.
%
%If a region $R$ stores information
%about $\Inf(p)$ for $p \not\in P(R)$ then 
%a bit indicating whether the region $R$ stores the
%originating region of a line segment in $\Left(R)$.
\item Finally, for each point $p \in P(R)$, we store the sequence of
bits from Lemma~\ref{lem:2sentropy}, which indicates whether $p$ has a non-empty 
$\Inf(p)$ and if so, for how many lines from $\Left(R) \cup \Inf(P(R))$, $p$ is 
a right endpoint ($p$ cannot be a right endpoint of
any other line in $\Inf(T)$, by the construction of the skeleton).
\end{enumerate}

As noted previously, the skeleton takes $O(n \sqrt{\log n})$ bits, within our
budget.
We now add up the space required for (1)-(5). By construction, the sum
of $|\Left(R)|$, $|\Right(R)|$ and $|P(R)|$ is $O(n)$ summed over all regions $R$.
The space bound for (1) and (3) is therefore $O(n)$ bits. 
The number of 1s in the bit string of (3), summed over all regions,
 is $O(n/\lambda)$, as there are
$O(n/\lambda)$ regions and the graph which indicates adjacency of regions is planar;
the number of $0$s is $O(n)$ as before.  The space used by (4) is $O(n \sqrt{\log n})$ bits again,
as for every $O(\sqrt{\log n})$ portions of line segments in the regions we store $O(\log n)$ bits.
Finally, the space used for (5) is $O(n)$ bits by Lemma~\ref{lem:2sentropy}.

\subsubsection*{Query algorithm.}
Suppose that we are given a query point $q$ in a sub-problem of size $n$ 
and need to answer $RMQ(q)$ (assume that we have $q$'s local and top-level
coordinates).  The query algorithm proceeds as follows:

\begin{enumerate}[label=(\alph*)]
\item Do a planar point location in the skeleton, and find a region $R$ in which 
the point $q$ lies.  Perform slab-select on $R$ to get $P(R)$.

\item As we know how many segments from
$\Left(R)$ lie vertically between any pair of points in $P(R)$, when we are given
the data in (5) above, we are able to determine whether the $x$-coordinate of
a given  point $p$ in $P(R)$ is the right endpoint of a line from either $\Left(R)$ or 
$\Inf(P(R))$. Thus, we have enough information to determine $\Inf(p)$ for all $p \in P(R)$ 
(at least until the right boundary of $R$).   
Furthermore, for each line in $\Left(R)$ that terminates in
$R$, we also know (the top-level coordinates of) its right endpoint.

\item Using the top-level coordinates of $q$, we determine
the nearest segment from $\Inf(P(R))$ that is above $q$.

\item  Using the top-level coordinates of $q$ we also find the set of segments from
$\Left(R)$ whose right endpoints are not to the left of $q$.  Let this set be 
$\Left^*(R)$. We now determine
the nearest segment from $\Left^*(R)$ that is above $q$. Unfortunately, although
$|\Left^*(R)| = O(\lambda)$, since the segments in $\Left^*(R)$ originate in points
outside $R$, we do \emph{not} have their $y$-coordinates.  Hence, we need to perform
the following binary search on $\Left^*(R)$:
\begin{enumerate}[label=(\alph{enumi}\arabic*)]
\item Take the line segment $\ell \in \Left^*(R)$ with median $y$-coordinate,
and suppose that $\ell = \Inf(p)$.
The first task is to find the region $R_p$ containing $p$, as follows. 
Use (2) to determine which of the adjacent regions
of $R$ $\ell$ intersects, say this is $R'$.  If $\ell$ ends in $R'$, 
or $R' = R_p$ and we are done. Otherwise, use (1) to locate
$\ell$ in $\Left(R')$ and continue.

\item Once we have found $R_p$, we perform a slab-select on $R'$ to  
determine $P(R_p)$, and sort $P(R_p)$ by $y$-axis.  Then we perform (c) above on
$P(R_p)$, thus determining which points of $P(R_p)$ have lines of influence that
reach the right boundary of $R_p$.  Using this we can now determine the (top-level)
coordinates of $p$.

\item We compare the top-level $y$-coordinates of $p$ and $q$ and recurse.
\end{enumerate}

\item We take the lower of the lines found in (d) and (e) and use it to return a candidate.
Observe that we have the top-level coordinates of this candidate.
\end{enumerate}

We now derive the time complexity of a 2-sided query.  
Step (a) takes $O(\log n)$ for the point location, and
$O(\log N/\log \log N)$ for the slab-select. 
%In step (b), finding $\hat{P}$ takes $O(\log n) = O(\log N)$ time by Lemma~\ref{lem:orr}. Narrowing it
%to $P(R)$ takes $O(\log N \log \log N)$ time; each
%slab-rank query takes $O(\log N)$ time as there
%are $O(\log \log N)$ slab-rank queries, each taking $O(\log N/\log \log N)$ time
%by Lemma~\ref{lem:orc}; this is done $O(\log \log N)$ times over the course of the binary search.
Step (b) can be done in $O(\log n) = O(\log N)$ time by running the plane sweep algorithm
of Lemma~\ref{lem:2sentropy} (recall that $|P(R)| = O(\sqrt{\log n})$---a 
quadratic algorithm will suffice).  Step (c) likewise
can be done by a simple plane sweep in $O(\log n)$ time.  Step (d1) is
iterated at most $O(\sqrt{\log n})$ times before $R_p$ is found
since every $\lambda$-th region intersected by $\ell$ contains information
about $p$.  Each iteration of (d1) takes $O(1)$ time: operations on the bit-strings are
done either by table lookup if the bit-string is short ($O(\lambda)$ bits), 
or else using rank and select operations \cite{DBLP:reference/algo/RahmanR08}, 
if the bit string is long (as e.g. the bit-string in (2) may be) -- these
entirely standard tricks are not described in detail.  Step (d2) takes
$O(\log N / \log \log N)$ time as before.  Steps (d1)-(d3) are performed
$O(\log \lambda) = O(\log \log N)$ times, so this takes $O(\log N)$ time overall.
Step (e) is trivial.  We have thus shown
Lemma~\ref{lem:2sfinal}.

\subsection{Putting things together}
\label{sec:combine-conclude}

As noted in Section~\ref{sec:space}, the space usage of our data structure is $O(N)$ words.
Coming to the running time, we solve $O(\log \log N)$ 2-sided queries using
Lemma~\ref{lem:2sfinal}, giving a time of $O(\log N \log \log N)$.  The 
$O(\log \log N)$ square-aligned queries are solved in $O(1)$ time each.  
The $O(1)$ terminal problems at the bottom of the recursion
are solved using Chazelle's algorithm in $O((\log \log N)^2)$ time. 
Any candidate given in local coordinates is converted to top-level coordinates
in $O(\log N/\log \log N)$ time,
or $O(\log N)$ time overall. We simply sequentially scan all $O(\log \log N)$ candidates
to find the answer. This proves Theorem~\ref{thm:4sfinal}.

\section{A succinct index for 2-sided queries}

In this section, we give a {succinct index} for 2-sided range maxima queries
over $N$ points in rank space. This is in essence a stand-alone variant of Lemma~\ref{lem:2sfinal} and reuses its
structure. As noted earlier, we consider the case where the point coordinates are
stored ``elsewhere'' and are assumed to be accessible 
in one of two ways, repeated here for convenience:

\begin{itemize}
\item Through an orthogonal range reporting query.  Here, we assume that a query
that results in $k$ points being reported takes $T(N,k)$ time.  We assume that
$T$ is such that $T(N, O(k)) = O(T(N,k))$.
\item The permuted-point model of \cite{Boseetal09}, where we assume that the point coordinates are stored
in read-only memory, permitting random access to the coordinates of the $i$-th point.
However, the ordering of the points is specified by the data structure.
\end{itemize}
Note that the priority information, unlike the coordinates, is encoded within the index. 

\subsection{Succinct index in the orthogonal range reporting model}

\begin{lemma}
\label{lem:index}
Let $\lambda \ge 2$ be some parameter.  There is a succinct index of size
$O(N + (N \log N)/\lambda))$ bits such that $RMQ(q)$ queries can 
be answered in $O(\log N + \log \lambda (\lambda + T(N, \lambda))$ time.
\end{lemma}
\begin{proof}
The proof follows closely the proof of Lemma~\ref{lem:2sfinal}, except that the
distinction between local and top-level coordinates vanishes, and the slab-select
operation is replaced by the assumed orthogonal range reporting query.  The space
complexity of the skeleton and associated bit-strings is exactly as in Lemma~\ref{lem:2sfinal}.
For the time complexity, the planar point location to find the region $R$
containing the query point $q$ is $O(\log N)$ time.  Finding $P(R)$ takes
$T(N, \lambda)$ time, and each of the $O(\log \lambda)$ iterations of the binary search
takes $O(\lambda + T(N,\lambda) )$ time.  \qed{}
\end{proof}
Choosing $\lambda = \log N$ in the above, we get:
\begin{corollary}
There is a succinct index of $O(N)$ bits such that $RMQ(q)$ queries can be answered in 
$O(\log \log N \cdot (\log N  + T(N,\log N)))$ time. 
\end{corollary}

\subsection{Succinct index in the permuted-point model}
 
\begin{lemma}
There is a succinct index of $N \log 5 + o(N) = 2.33N + o(N)$ bits for answering $RMQ(q)$ queries on
a set of $N$ points in $O(\log \log N)$ time in the permuted-point model.
\end{lemma}

\begin{proof} (sketch)
Let $S$ be the set of input points. We again solve $RMQ(q)$ queries by
planar point location on $\Inf(S)$.  
We consider the regions of the vertical decomposition
induced by $\Inf(S)$ as nodes in a planar graph (the dual graph of the
set of regions), with edges between adjacent cells
(see Fig~\ref{fig:verticaldecomp}).
\begin{figure}
\centering
\includegraphics[width=0.85\textwidth]{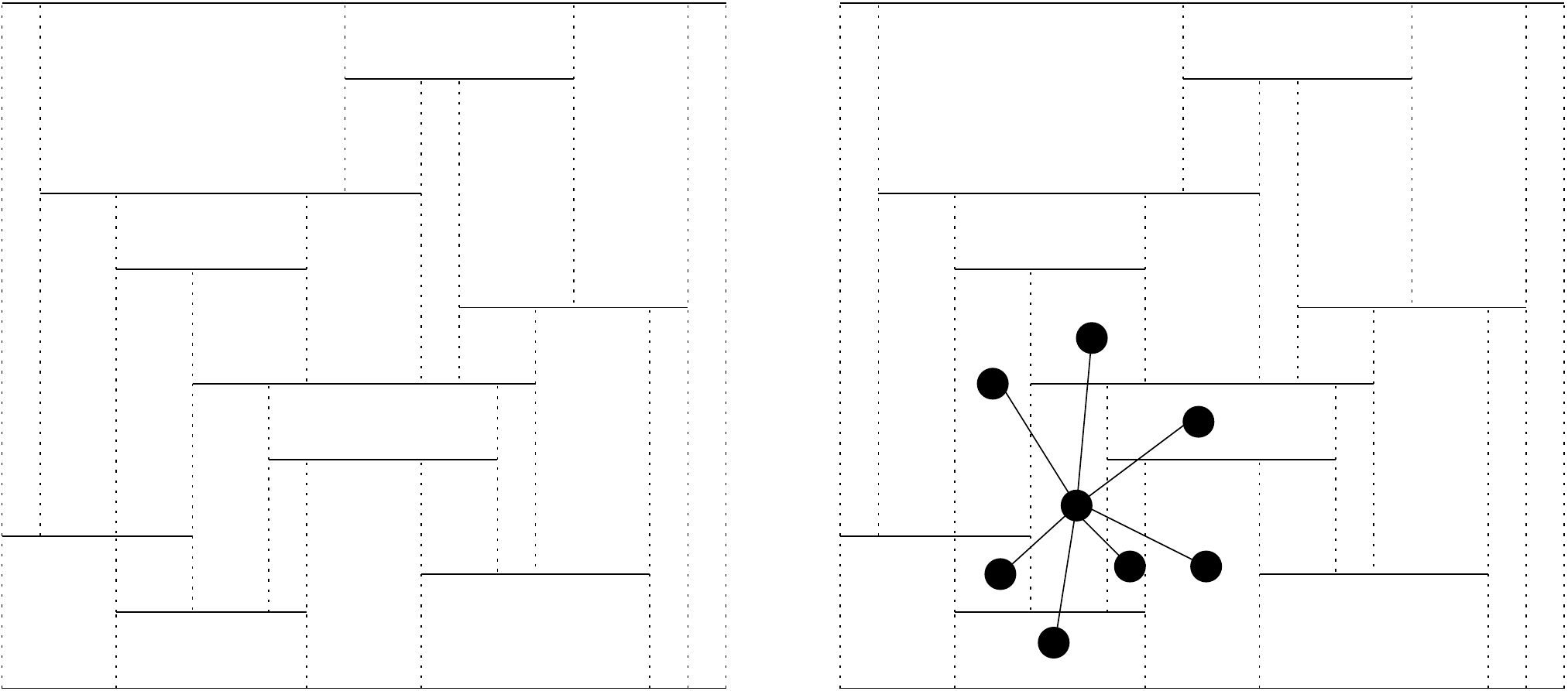}
\caption{Vertical decomposition of the plane induced by a collection of line segments, and a part of the dual graph.}
\label{fig:verticaldecomp}
\end{figure}
Using the planar separator theorem, for any parameter $\lambda \ge 1$, there
is a collection of $O(n/\lambda)$ cells such that the removal of this collection
of cells this graph can be decomposed into connected components
of size $O(\lambda^2)$ each (this approach is used by Bose \etalcite{Boseetal09} and
Chan and \Patrascu{} \cite{CP09} for example). As in \cite{Boseetal09} we use a
two-level decomposition, first decomposing the vertical decomposition of $\Inf(S)$ using
$\lambda = (\log N)^2$, and then decomposing the connected components themselves using
$\lambda' = (\log \log N)^2$.  The key difference to Lemma~\ref{lem:2sfinal} is that
the boundaries of the cells can be relatively compactly described. For instance,
for each separator cell in the top-level decomposition, we can specify the 
points in $S$ that define its four sides using $O(\log N)$ bits, and still use only $o(N)$ bits.
In the second-level decomposition, this information is stored in $O(\log \log N)$ bits per
separator cell (again $o(N)$ bits overall), 
as the relevant point will either belong to the same top-level connected
component of size $O((\log N)^4)$, or else the relevant information will be stored
in one of the top-level separator cells that form its boundary.  The leading-order
term comes from storing the $N \log 5 + o(N)$ bits of priority information needed to 
answer queries within the connected components.  As in \cite{Boseetal09},
we also permute the points in a manner aligned with the decomposition, 
which allows us to reconstruct
the appropriate part of the planar point location rapidly.  
The planar point location
in the top level is performed using Chan's data structure 
\cite{DBLP:conf/soda/Chan11} taking $O(\log \log N)$ time.
Note that a third level of decomposition, this time into connected components of
size $O((\log \log \log N)^3)$ is needed to achieve "decompression" of the
relevant parts of the point location structure in $O(\log \log N)$ time.
\qed{}
\end{proof}

\section{Conclusions}

We have introduced a new approach to producing space-efficient 
data structures for orthogonal range queries.  The main idea has been
to partition the problem into smaller sub-problems, which are stored
in a ``compressed'' format that are then ``decompressed'' on demand.
We applied this idea to give the first linear-space data structure for
2D range maxima that improves upon Chazelle's 1985 linear-space data structure.
%Although the running time is not very canonical, we believe that using
%relatively standard data structuring techniques it should be possible to reduce
%the exponent of the poly-log-log term in the running time.
\bibliographystyle{abbrv}
\bibliography{refs}
\appendix
\end{document}